\documentclass[11pt]{article}
\usepackage{fullpage}

\usepackage{times}
\usepackage{comment,amsfonts,amssymb,amsmath,amsthm,graphicx,algorithm,algorithmic}
\newcommand{\commentout}[1]{}

\ifx\pdftexversion\undefined
\usepackage[colorlinks,linkcolor=black,filecolor=black,citecolor=black,urlco
lor=black,pdfstartview=FitH]{hyperref}
\else
\usepackage[colorlinks,linkcolor=blue,filecolor=blue,citecolor=blue,urlcolor
=blue,pdfstartview=FitH]{hyperref}
\fi

\newcommand{\alert}[1]{\textbf{\color{red}
[[[#1]]]}\marginpar{\textbf{\color{red}**}}\typeout{ALERT:
\the\inputlineno: #1}}

\def\MathF{\hbox{\rm I\kern-2pt F}}
\def\MathP{\hbox{\rm I\kern-2pt P}}
\def\MathR{\hbox{\rm I\kern-2pt R}}
\def\MathZ{\hbox{\sf Z\kern-4pt Z}}
\def\MathN{\hbox{\rm I\kern-2pt I\kern-3.1pt N}}
\def\MathC{\hbox{\rm \kern0.7pt\raise0.8pt\hbox{\footnotesize I}
\kern-4.2pt C}}
\def\MathQ{\hbox{\rm I\kern-6pt Q}}

\newcommand{\diam}{{\rm diam}}
\newcommand{\rad}{{\rm rad}}

\newcommand{\N}{\mathbb{N}}

\newcommand{\R}{\mathbb{R}}

\newcommand{\E}{{\mathbb{E}}}

\newcommand{\mommit}[1]{}
\newcommand{\namedref}[2]{\hyperref[#2]{#1~\ref*{#2}}}
\newcommand{\sectionref}[1]{\namedref{Section}{#1}}

\newcommand{\theoremref}[1]{\namedref{Theorem}{#1}}

\newcommand{\figureref}[1]{\namedref{Figure}{#1}}

\newcommand{\claimref}[1]{\namedref{Claim}{#1}}
\newcommand{\lemmaref}[1]{\namedref{Lemma}{#1}}

\newcommand{\corollaryref}[1]{\namedref{Corollary}{#1}}

\newtheorem{theorem}{Theorem}
\newtheorem{lemma}{Lemma}
\newtheorem{corollary}[lemma]{Corollary}

\newtheorem{claim}[lemma]{Claim}

\newtheorem{definition}{Definition}

\usepackage{pdfsync}
\usepackage{authblk}

\def\eps{\epsilon}

\title{Lossless Prioritized Embeddings}
\usepackage{authblk}
\usepackage{pdfsync}
\author[1]{Michael Elkin}
\author[1]{Ofer Neiman}
\affil[1]{Ben-Gurion University of the Negev. Email: \texttt{\{elkinm,neimano\}@cs.bgu.ac.il}}

\date{}
\begin{document}
\maketitle
\begin{abstract}
Given metric spaces $(X,d)$ and $(Y,\rho)$ and an ordering $x_1,x_2,\ldots,x_n$ of $(X,d)$, an embedding $f: X \rightarrow Y$ is said to have a {\em prioritized distortion} $\alpha(\cdot)$, for a function $\alpha(\cdot)$, if for any pair $x_j,x'$ of distinct points in $X$, the distortion provided by $f$ for this pair is at most $\alpha(j)$. If $Y$ is a normed space, the embedding is said to have {\em prioritized dimension} $\beta(\cdot)$, if $f(x_j)$ may have  non-zero entries only in its first $\beta(j)$ coordinates.

The notion of prioritized embedding was introduced by Filtser and the current authors in \cite{EFN15}, where a rather general methodology  for constructing such embeddings was developed. Though this methodology enables \cite{EFN15} to come up with many prioritized embeddings,
it typically incurs some {\em loss} in the distortion. In other words, in the worst-case, prioritized embeddings obtained via this methodology incur distortion which is at least a constant factor off, compared to the distortion of the classical counterparts of these embeddings. This constant loss is problematic for isometric embeddings. It is also troublesome for Matousek's embedding of general metrics into $\ell_\infty$, which for a parameter $k = 1,2,\ldots$, provides distortion $2k-1$ and dimension $O(k \log n \cdot n^{1/k})$.

In this paper we devise two {\em lossless} prioritized embeddings. The first one is an {\em isometric prioritized} embedding  of tree metrics into $\ell_\infty$ with dimension $O(\log j)$, matching the worst-case guarantee of $O(\log n)$ of the classical embedding of Linial et al. \cite{LLR95}. The second one is a prioritized Matousek's embedding of general metrics into $\ell_\infty$, which for a parameter $k=1,2,\ldots$, provides prioritized distortion $2 \lceil k {{\log j} \over {\log n}} \rceil - 1$ and dimension $O(k \log n \cdot n^{1/k})$, again matching the worst-case guarantee $2k-1$ in the distortion of the classical Matousek's embedding.

We also provide a dimension-prioritized variant of Matousek's embedding. Finally, we devise prioritized embeddings of general metrics into (single) ultra-metric and of general graphs into (single) spanning tree with asymptotically optimal distortion.

\end{abstract}
\section{Introduction}

An injective function $f: X \rightarrow Y$ between the metric spaces $(X,d)$ and $(Y,\rho)$ is called an {\em embedding}. The embedding is said to have {\em distortion} at most $\alpha$, if for for every pair $x,x' \in X$ of distinct points we have
$1 \le \frac{\rho(f(x),f(x'))}{d(x,x')} \le \alpha$.\footnote{This definition applies actually to {\em non-contracting} embeddings. For a {\em non-expanding} of {\em Lipschitz} embedding, one requires $1/\alpha \le \frac{\rho(f(x),f(x'))}{d(x,x')} \le 1$. Yet more generally, there may also be some normalizing term $\eta$, so that $\eta \le \frac{\rho(f(x),f(x'))}{d(x,x')} \le \alpha \cdot \eta$.
The quotient $\frac{\rho(f(x),f(x'))}{d(x,x')}$ is called the {\em distortion} of the pair $x,x'$ under the embedding $f$.}

Study of low-distortion embeddings, initiated with seminal works of \cite{JL82,A83,B85} in the eighties has flourished since then, and has found numerous algorithmic applications. See, e.g., \cite{LLR95,B96,B98,FRT04,GKL03,LMN04,ABN06,EFN15,EFN15t,BFN16,AFGN18}, and the references therein. Closely related to it is the study of metric structures, such as graph spanners and low-stretch spanning trees, distance labeling schemes and distance oracles
\cite{PS89,AKPW95,EP04,EEST05,TZ01,T01,Peleg99,CE06,AB17,EN19}, and the references therein.

In all known metric embeddings either the distortion or the dimension of the target metric space (or, in most cases, both) depend at least logarithmically on the cardinality $n = |X|$ of the point set, and moreover, this dependence is typically unavoidable \cite{LLR95,ABN06}.
This dependence was always viewed as an inherent downside of the entire approach. To alleviate the issue, Filtser and the current authors \cite{EFN15} introduced\footnote{There were earlier works \cite{KSW04,ABCD05,ABN06} that studied partial and scaling embeddings, that provide guarantees in terms of a parameter $\eps > 0$, for all pairs $(x,x')$ such that $x'$ is not among the $\eps\cdot n$ closest points to $x$. However, these embeddings do not allow one to select for which pairs of points a better distortion is provided.} the notions of {\em prioritized} embeddings and metric structures.
For some function $\alpha(\cdot)$, given an ordering $x_1,x_2,\ldots,x_n$ of the points in the source metric space $X$, an embedding $f: X \rightarrow Y$ is said to be {\em prioritized} with {\em distortion} $\alpha(\cdot)$, if it guarantees that each pair $(x_j,x')$ of distinct points of $X$ admits distortion at most $\alpha(j)$ (as opposed to $\alpha=\alpha(n)$ guaranteed by classical, non-prioritized embeddings). In the case that $Y$ is a normed space, it is said to be {\em prioritized} with {\em dimension} $\beta(\cdot)$, if the number of nonzero coordinates in $f(x_j)$ is at most $\beta(j)$, and moreover, these nonzero coordinates are the first $\beta(j)$ coordinates of $f(x_j)$.

Even though these notions might look too demanding, surprisingly \cite{EFN15} showed that they are achievable for numerous settings.
In particular, a prioritized analogue of Bourgain's embedding \cite{B85,ABN08} shown there states that any $n$-point metric space can be embedded in $\mathbb{R}^{O(\log n)}$ with prioritized distortion $O(\log j \cdot (\log\log j)^{1/2+\eps})$, for an arbitrarily small constant $\eps >0$. This was later refined in \cite{BFN16} to the optimal bound of $O(\log j)$. Moreover, it was shown in \cite{EFN15} that one can achieve prioritized distortion {\em and} dimension of $\mathit{poly}(\log j)$ simultaneously in embedding general metrics into Euclidean ones.

Another notable embedding shown in \cite{EFN15} is an embedding of general metrics into a (single) tree metric. For the classical scenario a tight bound of $\Theta(n)$ for such embeddings is long well-known \cite{RR98}. In \cite{EFN15} Filtser and the current authors devised for this scenario a prioritized embedding with distortion $O(\alpha(\cdot))$ for any function $\alpha(\cdot)$ that satisfies
\begin{equation}
\label{eq:converg_cond}
\sum_{j=1}^\infty {1 \over {\alpha(j)}} < 1~.
\end{equation}
Moreover, it was shown there that the condition (\ref{eq:converg_cond}) is necessary, i.e., if $\alpha(\cdot)$ does not satisfy the condition, no such embedding with prioritized distortion at most $\alpha(\cdot)/8$ is possible. Observe that linear functions do not satisfy the condition
 (\ref{eq:converg_cond}), and thus a certain {\em loss} when extending classical embeddings to prioritized setting {\em must} in general be incurred. (Observe also that $\alpha(n) = c\cdot n \cdot \log n (\log\log n)^{1.01} $ does satisfy  (\ref{eq:converg_cond}) for a constant $c$, and thus the loss in this case is not too large, i.e., it is at most polylogarithmic in the original classical distortion.)
On the other hand, for the FRT embedding of general metrics into ultrametrics  \cite{B96,B98,FRT04}, it was shown in \cite{EFN15} that an embedding with prioritized distortion $O(\log j)$ exists, i.e., in this case the loss is at most some constant.

In fact, \cite{EFN15} presented a rather general methodology for converting classcial embeddings and metric structures into prioritized ones (with respect to distortion). Roughy speaking, the methodology consists of two parts. In the first part one devises a {\em strong  terminal} analogue of the classical embedding.\footnote{We say that an embedding $f: X \rightarrow Y$ has {\em terminal} distortion $\alpha(\cdot)$ with respect to a subset $K \subseteq X$, $|K| = k$, of $k$ terminals, if it guarantees distortion $\alpha(k)$ for all pairs $(x,x') \in K \times X$. The embedding has {\em strong terminal} distortion, if it has terminal distortion, and in addition it has distortion $\alpha(n)$ for all pairs of points.} In the second part, one partitions the point set $X$ into a collection of disjoint subsets $X = \bigcup X_i$ according to the given priority ordering, where the set sizes $|X_1|, |X_2|,\ldots$ grow quickly.\footnote{The specific growth rate depends on the embedding at hand.} One then builds strong terminal embeddings $f_i$ with respect to each of the sets $X_i$ as terminal sets, and carefully combines these embeddings into the ultimate prioritized embedding $f$.

Typically both these steps incur some loss in distortion. The loss incurred by the first part of this scheme is at most constant, under very general conditions described in \cite{EFN15t}. The loss incurred in the second part is not yet well-understood. Specifically, in the prioritized variant of Bourgain's embedding \cite{EFN15} it is $O((\log\log j)^{1/2+\eps})$, for an arbitrarily small $\eps > 0$,\footnote{The improvement in \cite{BFN16} was achieved by a different approach.} while in the prioritized variant of the FRT embedding \cite{EFN15} the loss is at most a constant.

The constant loss in the first part of the scheme has to do with an outer extension, implicitly developed in \cite{EFN15t}, and explicated in \cite{MMMR18,NN18}. Bi-Lipschitz outer extensions have been a focus of recent research \cite{MMMR18,NN18,EN18}, where they were studied in the context of Johnson-Lindenstrauss dimension reduction \cite{MMMR18,NN18}, and in the context of doubling metrics \cite{EN18}. In both these contexts it was shown that the loss can be made at most $1+ \eps$, for an arbitrarily small $\eps > 0$. However, in general, Mahabadi et al. \cite{MMMR18} showed that it is unfortunately not the case, as there are scenarios when Bi-Lipschitz outer extensions incur loss at least $c$, for some  constant $c > 1$.

Generally speaking, the loss is most troublesome in two scenarios.  The first one is when the embedding is isometric or near-isomentric (i.e., has distortion $1+\eps$, for an arbitrarily small $\eps>0$), as it is the case for the Johnson-Lindenstrauss dimension reduction \cite{MMMR18,NN18} and for embedding of doubling metrics into $\ell_\infty$ \cite{EN18}. The second scenario is when the specific leading constant factor in the distortion is important. Consider Matousek's embedding \cite{Mat96b} of general metrics into $\ell_\infty$. Given an integer parameter $k = 1,2,\ldots$, this embedding has distortion $2k-1$ and dimension $O(k n^{1/k} \cdot \log n)$. Employing the generic Lipschitz outer extension of \cite{EFN15t} on this embedding results in a prioritized embedding with distortion $4 \lceil {{k\log j} \over {\log n}} \rceil-1$ and dimension
$O(k^2\cdot n^{1/k}\cdot \log n)$.  (See \sectionref{sec:mat-prior} for more details.) In other words, in this case there is a significant loss both in the distortion and in the dimension. (Note that for distortion $4k-1$ the classical Matousek embedding provides dimension
$O(k \cdot n^{1/2k} \log n)$ , i.e., almost {\em quadratically smaller} than the dimension obtained via the generic methodology.)

In this paper we address both these scenarios. We devise an isometric embedding of tree metrics into $\ell_\infty$ with {\em prioritized  dimension} $O(\log j)$. Note that in the context of {\em isometric} embeddings, no loss whatsoever can be tolerated. This is unlike the case of near-isometric embeddings, where one can afford some loss in the distortion. This is indeed the case in the terminal Johnson-Lindenstrauss dimension reduction of Mahabadi et al. \cite{MMMR18} and Narayanan and Nelson \cite{NN18}, where the distortion grows from
$1+ \eps$ to $1 +O(\sqrt{\eps})$ in \cite{MMMR18} and to $1 + O(\eps)$ in \cite{NN18}.  The loss is similar (i.e., $1+ \eps$ becomes $1 + O(\eps)$) in the terminal embedding of \cite{EN18} (by the current authors) of doubling metrics to $\ell_\infty$.
Note also that the aforementioned embeddings \cite{MMMR18,NN18,EN18} are {\em not prioritized}\footnote{It is moreover open if these embeddings can be made prioritized. This question was explicitly studied in \cite{MMMR18}, where a prioritized version of JL dimension reduction was shown, albeit with a much larger (than $1+\eps$) distortion of $O(\log\log j)$,  and under a weaker notion of prioritized distortion than the one we study here.}, while our embedding of tree metrics into $\ell_\infty$ is.
This embedding of ours is both isometric  and {\em lossless prioritized}, while previously there were no known\footnote{While preparing this submission we were informed \cite{FGK19} that a similar result was lately achieved independently of us by Filtser, Gottlieb and Krauthgamer.} strong terminal or prioritized isometric embeddings. Moreover, the prioritized dimension $O(\log j)$ of our embedding is optimal \cite{matbook}.

Our second contribution addresses the second scenario discussed above. Specifically, we devise a {\em lossless prioritized} variant of Matousek's embedding. We show that for any positive integer parameter $k = 1,2,\ldots$, there is an embedding of general metrics into $\ell_\infty$ with prioritized distortion $2 \lceil {{k\log j} \over {\log n}} \rceil  - 1$ and dimension $O(k \cdot  n^{1/k} \cdot \log n)$. Observe that by substituting $j= n$ one obtains the parameters of classical Matousek's embedding {\em without any loss} in distortion, and the dimension remains asymptotically the same.

We also devise a dimension-prioritized variant of Matousek's embedding. Specifically, this variant of our embedding, for any positive integer parameter $k = 1,2,\ldots$, provides distortion $O(k \log\log j)$ and prioritized dimension  $O(k \cdot j^{1/k} \cdot \log n)$.
(This embedding does  leave some room for improvement. One can hope to obtain distortion $2k-1$ and dimension as above, but with $\log n$ replaced by $\log j$.)

Finally, our third contribution has to do with the aforementioned prioritized embedding of \cite{EFN15} of general metrics into  (single) tree metrics.
As was discussed above, this embedding provides prioritized dimension $\alpha(\cdot)$ for any function $\alpha$ that satisfies
the condition (\ref{eq:converg_cond}), and this condition is tight. In our opinion, this embedding is a Drosophila of prioritized embeddings, i.e., it is fundamental to the area, and deserves further exploration.
In this paper we devise an embedding of general metrics into a (single) ultrametric (see \sectionref{sec:single-ultra} for the definition), and an embedding of general graphs into a (single) spanning tree.
Both these embeddings have the same (tight up to constant factors) prioritized distortion $O(\alpha(\cdot))$, for any function $\alpha(\cdot)$ that satisfies the condition (\ref{eq:converg_cond}). (The tightness follows from \cite{EFN15}.)

\subsection{Technical Overview}

\paragraph{Embedding General Metrics into $\ell_\infty$.}
The embedding of Matousek \cite{Mat96b} is a Frechet one, i.e., for every index $i$, the $i$th coordinate of an image $f(x)$ of an input point $x \in X$ is a distance from $x$ to a certain subset $S_i \subseteq X$. The set $S_i$ (henceforth, a {\em Frechet} subset) is constructed by sampling points from $X$ independently at random with appropriate probabilities. In the analysis of distortion for a fixed pair of points $x,x' \in X$ such that $d(x,x') = d$, one sets
$\delta = {d \over {2k-1}}$. One then considers pairs of balls $(B_0 = B(x,0 \cdot \delta),B'_1=B(x',1\cdot \delta)),(B'_1,B_2= B(x,2\cdot \delta)),\ldots,(B_{k-1},B'_k = B(x',k \cdot \delta))$.
Intuitively, for some index $i \le k-1$, the ratio between the sizes of $B'_{i+1}$ and $B_i$ is at most $n^{1/k}$, and at this point, with high probability, one of the two balls hits a Frechet set $S_j$ of points selected with sampling probability $n^{-(1 - i/k)}$, while the other ball does not. Therefore, in the respective coordinate $f(x)$ and $f(x')$ (where $f(\cdot)$ is the ultimate embedding) will differ by at least $\delta$, ensuring that the distortion is at  most $2k-1$.

In our prioritized variant  of this embedding we carefully design sampling probabilities so that higher-priority points are more likely to be included in the sets $S_j$ than lower-priority ones. In the analysis of distortion for a pair $(x,x')$ that contains a high-priority point $x$, we build a different, much shorter sequence of ball-pairs $(B_i,B'_{i+1}),(B'_{i+1},B_{i+2}),\ldots,(B_{k-1},B'_k)$, where the index $i$ depends on the priority of $x$. The balls now will have radii that are multiples of $\delta_i = {d \over {2(k-i) - 1}}$, as opposed to $\delta = {d \over {2k-1}}$ used in the classcial embedding of Matousek. In other words, while Matousek's analysis of distortion for a pair $(x,x')$ with $d(x,x') = d$ considers balls of radii $0,{d \over {2k-1}},{{2d} \over {2k-1}},\ldots,{{k-1} \over {2k-1}} d$, our analysis of the distortion considers a sequence of radii that depends not just on the distance between $x$ and $x'$, but also on the priorities of these points.
Since this sequence of radii is shorter than $k$ for pairs containing at least one high-priority point, the resulting distortion  is smaller than $2k-1$. Moreover, for points with very high priorities (i.e., priorities at most $n^{1/k}$), the resulting distortion between them and {\em all other points} is exactly $1$. For points with priorities between $n^{1/k}$ and $n^{2/k}$, all pairs involving them admit distortion at most 3, etc..

Our approach also leads to {\em dimension}-prioritized
variant of Matousek's embedding. As high-priority points are more likely to be included in the Frechet sets, they end up having more zeroes in the respective coordinates. However, the argument becomes insufficient when the sampling probabilities are very small, e.g., $1/n$.
So we need to have only very few Frechet sets with such a small sampling probability (rather than $\approx n^{1/k}$ such sets in the standard embedding). In order to allow the distortion analysis to go through, the idea is to refine the sampling probabilities, and introduce $\log\log n$ scales that have smaller "gaps" between consecutive levels as the probabilities get closer to $1/n$. This increased number of scales translates to the (prioritized) factor of $\log\log j$ in the distortion.

\paragraph{Embedding of trees to $\ell_\infty$ with prioritized dimension.}
The isometric embedding of a tree into $\ell_\infty$ \cite{LLR95} is based on centroids: every tree $T$ has a vertex $s$, such that there are two trees $T_1$, $T_2$ who cover $T$ and share only $s$, each of size at most $2|T|/3$. One can define an embedding to $\R$ that preserves exactly distances between each $u\in T_1$ to every $v\in T_2$, and then embed the two trees recursively.

How can we guarantee that high priority points will receive few nonzero coordinates in such an embedding? We can try to follow the \cite{EFN15} methodology described above: partition the vertex sets into few terminal sets $K_1,K_2,\dots$, and apply a strong terminal embedding on each $K_i$. However, it is not clear how to combine the different embeddings without incurring large dimension $\Theta(\log n)$ for all points. One possible attempt is to contract all the vertices in $K_i$ to a single vertex $z$, and only then apply the embedding for $K_{i+1}$, ensuring that $z$ will be mapped to 0 in that embedding. The issue is that contraction may completely change the distances.

To handle this problem, we introduce a novel graph operation, that we call {\em folding}. Rather than contracting $K_i$, we iteratively "fold" the vertices in $K_i$ one to the other, keeping half of the shortest path between them intact. While this operation can also change distances, we prove a structural lemma characterizing which pairs may be affected, and ensure that such pairs are handled by the previous level embedding.

\paragraph{Priority embedding of arbitrary metrics into a single tree.}
The embedding of \cite{EFN15} is based on re-weighting the input metric (or graph), and finding its MST. Our embedding into ultrametrics is based on a different approach: we use hierarchical ball partitions, and in each iteration we find a suitable ball that does not separate points if their distortion will be too large. We show that one can find such a ball which is simultaneously good for all possible priorities. In the case of spanning trees, we apply the Petal Decomposition framework of \cite{AN12}, combined with this ball partition.


\section{Prioritized Embedding of Trees Into $\ell_\infty$}

\subsection{Path Folding}
For a graph $G=(V,E)$, {\em identifying} two vertices $x,y\in V$ means replacing them by a single vertex $z$, which is connected to every neighbor of $x$ and of $y$ (keeping only the shorter edge among parallel edges). We often abuse notation and denote the vertex set of the new graph obtained after this operation by $V$ as well, that is, we may refer to the new vertex $z$ as either $x$ or $y$.

Given a graph $G=(V,E)$, two vertices $u,v\in V$, and some shortest path $P$ from $u$ to $v$ of length $t$, for each $0\le\alpha\le 1$ let $x_\alpha$ be the (possibly imaginary) vertex at distance $\alpha\cdot t$ from $u$ on the path $P$. (So that $x_0=u$ and $x_1=v$.)  Define the {\em folding} of $P$ as the operation that produces a graph $G'$ out of $G$ by identifying, for every $0\le\alpha\le 1/2$, the vertex $x_\alpha$ with the vertex $x_{1-\alpha}$. The vertex $x_{1/2}$ is called the {\em folding point}. See \figureref{fig:folding}.

\begin{figure}[t]
	\centering{\includegraphics[scale=0.62]{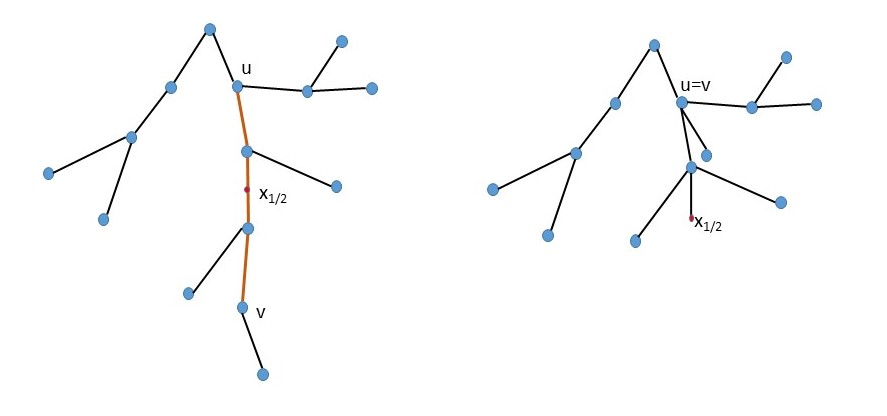}}
	\caption{\label{fig:folding} \small  An example of path folding. In the left is a tree with the $u-v$ path marked, on the right is the tree obtained by folding this path, with $x_{1/2}$ the folding point.}
\end{figure}

\begin{claim}\label{claim:remain-tree}
Let $T$ be a tree, then for any path $P$ in $T$, if $T'$ is obtained from $T$ by folding $P$, then $T'$ is a tree as well.
\end{claim}
\begin{proof}
It is easy to see that $T'$ is connected (identifying vertices cannot disconnect a graph).
Seeking contradiction, let $C'$ be a cycle in $T'$. Clearly $C'$ must contain some of the vertices $\{x_\alpha\}$, as otherwise $C'$ would be a cycle in $T$. We now create a cycle in $T$, and derive a contradiction. For every (non-imaginary) vertex $x_\alpha$ in $C'$: if both edges of $C'$ touching it are edges of $T$ that touch $x_\alpha$ (respectively, $x_{1-\alpha}$), we do nothing. Otherwise, we replace $x_\alpha$ by the subpath of $P$ from $x_\alpha$ to $x_{1-\alpha}$. After all these replacements, we get a cycle in $T$.
\end{proof}

We also consider multiple foldings: given a set $K\subseteq V(T)$ of designated vertices in the tree $T$, a {\em $K$-folding} generates a tree $\hat{T}$ in which all vertices of $K$ are identified as the same vertex. This is done by iteratively selecting two (arbitrary) vertices in $K$ that are not identified yet, and folding them along the unique shortest path, so that they become identified. Observe that \claimref{claim:remain-tree} implies that $K$-folding indeed generates a tree, but that tree is not unique.

For $u,v\in T$ let $P_{uv}$ be the unique path between them in $T$. The purpose of the next claim is to show that folding a path in a tree can change distances only for pairs $u,v$ whose path $P_{uv}$ in $T$ crosses (non-trivially) the folding point. For a path $P$ with a folding point $x=x_{1/2}$, we denote by $P^-$ the first half of the path: the interior of the path from $x_0$ to $x$. Denote by $P^+$ the interior of the path from $x$ to $x_1$. We say that a pair of vertices $u,v$ {\em crosses} $x$, if $P_{uv}$ intersects both $P^-$ and $P^+$.
\begin{claim}
Let $T$ be a tree, and $T'$ the tree obtained by folding some path $P$. If $x=x_{1/2}$ is the folding point, then for any $u,v\in T$ that do not cross $x$, we have
\[
d_T(u,v)=d_{T'}(u,v)
\]
\end{claim}
\begin{proof}
We first observe that identifying vertices is equivalent (distance-wise) to adding a zero-weight edge between them. So if $G'$ is the graph obtained from $T$ by adding zero weight edges between every vertex $x_\alpha$ to $x_{1-\alpha}$, it follows that $d_{T'}(y,z)= d_{G'}(y,z)\le d_T(y,z)$ for every $y,z\in V$. Clearly, if $P_{uv}$ does not intersect $P$ then $d_{T'}(y,z)=  d_T(y,z)$. Otherwise, let $\alpha_1\le\alpha_2\le 1$ be the first (resp., last) index of a vertex of $P_{uv}$ on $P$. By our assumption, either both $\alpha_1,\alpha_2\le 1/2$ or both $\alpha_1,\alpha_2\ge 1/2$.
The main observation is that the zero-weight edges added to $T$ will not shorten the distance between $x_{\alpha_1}$ to $x_{\alpha_2}$ (i.e. $d_T(x_{\alpha_1},x_{\alpha_2})=d_{T'}(x_{\alpha_1},x_{\alpha_2})$), this is because the folded path is a shortest path. The claim follows.
\end{proof}

\begin{corollary}\label{cor:maintain-distance}
If $T$ is a tree and $\hat{T}$ a $K$-folding of $T$ that uses $k-1$ folding points $X=\{x^{(1)},\dots,x^{(k-1)}\}$, then for any pair $u,v\in T$ that do not cross any point in $X$, we have $d_{\hat{T}}(u,v)=d_T(u,v)$.
\end{corollary}

\subsection{Prioritized Embedding}

We now define an isometric embedding of $T$ into $\ell_\infty$ with prioritized dimension $O(\log j)$. That is, given a ranking $(v_1,\dots,v_n)$ of the vertices of $T$, vertex $v_j$ will be nonzero in at most $O(\log j)$ coordinates. The main technical lemma provides, for a given set of terminals $K\subseteq V$, two basic constructions: 1) A {\em folding} of the tree $T$ along $|K|-1$ paths, so that in the tree $\hat{T}$ created by these foldings, all terminals are identified as a single vertex. 2) An embedding which is isometric on every pair whose shortest path in $T$ crosses a folding point.

We will use the following well-known claim (see, e.g., \cite{LLR95}).
\begin{claim}\label{claim:tree-separator}
Let $T=(V,E)$ be a tree, and $K\subseteq V$ a set of vertices. Then there exist subtrees $T_1$
and $T_2$ of $T$ that share a single vertex and together cover $T$, such that each of $T_1$, $T_2$ contain at most $\lceil 2|K|/3\rceil$ vertices of $K$.
\end{claim}

\begin{lemma}\label{lem:embed}
Let $T=(V,E)$ be a tree, and $K\subseteq V$ a set of $k$ terminals. Then there exists a non-expanding embedding $f:T\to\ell_\infty^{O(\log |K|)}$ and a tree $\hat{T}$ which is a $K$-folding of $T$, so that for every $u,v\in V$ at least one of the following holds:
\begin{enumerate}
\item $\|f(u)-f(v)\|_\infty = d_T(u,v)$.
\item $d_{\hat{T}}(u,v)= d_T(u,v)$.
\end{enumerate}
\end{lemma}
\begin{proof}
Let $X$ be the set of $k-1$ folding points in some fixed $K$-folding of $T$
.
By \corollaryref{cor:maintain-distance} it suffices to show a non-expansive embedding $f:T\to\ell_\infty^{O(\log k)}$ that preserves exactly distances between pairs $u,v\in V$ whose shortest path crosses a folding point in $X$.
We assume w.l.o.g. that the path in $T$ between any two distinct points $x,y\in X$ contains a vertex in its interior (if it does not, add an imaginary vertex on the edge $(x,y)$).
Consider the $2(k-1)$ half paths $P_1^-,P_1^+,\dots,P_{k-1}^-,P_{k-1}^+$ in the $K$-folding of $T$. We prove by induction on $b$, that if $H$ is a subtree of $T$ with $b=|H\cap X|$, then it admits a non-expansive embedding into $\R^{4(\log b+1)}$ with distortion 1 (in the $\ell_\infty$ norm) for all pairs whose path in $H$ crosses (at least) one of these $b$ vertices. Clearly a path in $H$ cannot cross a vertex which is not in $H$.

The base case is when $b=1$ (when $b=0$ we can just map all vertices to $0$). Let $x\in X$ be the single point in $H\cap X$, and let $P^-$ and $P^+$ be the two parts of the path $P$ for which $x$ is the middle point. We note that for any vertex $u\in H$, the unique path from $x$ to $u$ can intersect at most one of $P^-$ or $P^+$ (since both of these paths emanate from $x$).
Let $R\subseteq H$ be the set of vertices $u\in H$ such that the path from $x$ to $u$ intersects $P^-$, and define the embedding $f:H\to\R$ by
\[
f(u)=\left\{\begin{array}{ccc} d_H(x,u) & u\in R\\
-d_H(x,u) & \text{otherwise} \end{array} \right.
\]
We first observe that if $u\in R$ and $v\notin R$, then the path between $u,v$ must contain $x$. This is because the very first edge $e$ on the path from $x$ to $u$ is on $P^-$, so if the path from $u$ to $v$ does not pass through $x$, then the path from $x$ to $v$ would contain $e$, a contradiction to the fact that $v\notin R$. This observation implies that $f$ is non-expansive: for pairs with both $u,v\in R$ or both $u,v\notin R$ by the triangle inequality, and for pairs $u\in R$, $v\notin R$ we have $d_H(u,v)=d_H(u,x)+d_H(x,v)=|f(u)-f(v)|$. The latter proves that pairs $u,v$ whose path crosses $x$ (recall the path must intersect both $P^-$ and $P^+$) attain distortion 1, as required.

For the induction step, let $H_1$ and $H_2$ be the subtrees of $H$ obtained by applying \claimref{claim:tree-separator} on the tree $H$ with respect to the terminal set $H\cap X$ (we stress that this terminal set is not the given terminal set $K$), and let $s\in H$ be the unique vertex common to both subtrees.
Let $b_1=|X\cap H_1|$ and $b_2=|X\cap H_2|$, by \claimref{claim:tree-separator} we have that both $b_1,b_2\le b'=\lceil 2b/3\rceil$. (Note that in the case $b=2$ we have $\lceil 2b/3\rceil=2$ as well, but since we ensured that there is at least one vertex on the path between any two points of $X$, this vertex can serve as a separator, and each $H_1,H_2$ will contain only one point of $X$.) Let $f_1:H_1\to\R^{4\log b'+1}$ and $f_2:H_2\to\R^{4\log b'+1}$ be the embeddings guaranteed by the induction hypothesis. We may shift these embeddings so that $f_1(s)=f_2(s)=0$, without affecting the distortion. Note that for $b>2$ we have that $4\log \lceil 2b/3\rceil +1\le 4\log b$ (and for $b=2$, we have that both $H_1,H_2$ contain a single vertex of $X$, so by the base case, $f_1$ and $f_2$ are embeddings into $\R$). Thus we can define $f:H\to\R^{4\log b+1}$ by adding a single coordinate $h:H\to\R$ with
\[
h(u)=\left\{\begin{array}{ccc} d_H(s,u) & u\in H_1\\
-d_H(s,u) & u\in H_2 \end{array} \right.
\]
Define the embedding for $u\in H_i$, $i\in \{1,2\}$, by $f(u)=f_i(u)\oplus h(u)$. We claim that $f$ in indeed non-expansive. To see this, consider first a pair $u,v\in H_i$, then by the inductive hypothesis $f_i$ is non-expansive, and the additional coordinate $h$ is also non-expansive by the triangle inequality. Consider now a pair $u\in H_1$ and $v\in H_2$. As before we have that $d_H(u,v)=d_{H_1}(u,s)+d_{H_2}(s,v)$. Since $s$ belongs to both trees and $f_1(s)=f_2(s)=0$, it follows that
\begin{eqnarray*}
\|f_1(u)-f_2(v)\|_\infty &=& \|f_1(u)-f_1(s)+f_2(s)-f_2(v)\|_\infty\\ &\le&\|f_1(u)-f_1(s)\|_\infty+\|+f_2(s)-f_2(v)\|_\infty\\&\le& d_{H_1}(u,s)+d_{H_2}(s,v)=d_H(u,v)~.
\end{eqnarray*}
For the additional coordinate $h$ we have $|h(u)-h(v)|=d_H(u,s)+d_H(s,v)=d_H(u,v)$, so $f$ is indeed an isometry for such $u,v$. By the induction hypothesis, every pair $u,v\in H_i$, $i\in \{1,2\}$, for which the path from $u$ to $v$ crosses some $x\in H\cap X$, has distortion 1 by $f_i$, and thus also by $f$.
\end{proof}

\begin{theorem}
Let $T$ be a tree on $n$ vertices with priority ordering $(v_1,\dots,v_n)$. Then there is an isometric embedding of $T$ into $\ell_\infty$ with prioritized dimension $O(\log j)$.
\end{theorem}
\begin{proof}

We partition the vertices of $T$ into $\lceil\log\log n\rceil$ sets $K_1,\dots,K_{\lceil\log\log n\rceil}$ by taking $K_i=\{v_j~:~ 2^{2^{i-1}}<j\le 2^{2^i}\}$ (put $v_1,v_2$ into the first set, and the last set contains vertices only until $v_n$). Initially, let $T_1=T$, $z=v_1$, and in the $i$-th step, apply the embedding of \lemmaref{lem:embed} on the tree $T_i$ with the set $K_i\cup\{z\}$. Define $T_{i+1}=\hat{T}_i$ the tree of the next iteration as the $K_i\cup\{z\}$-folding of $T_i$. Let $f_i:T_i\to\R^{O(2^i)}$ be the embedding of $T_i$, we may assume by an appropriate shift that $f_i(z)=0$, note that $z\in T_i$ is the unique vertex corresponding to the terminals $v_1,\dots v_{2^{2^{i-1}}}$. The embedding is defined as $f=\bigoplus_if_i$.

\paragraph{Isometry.} To see that $f$ is an isometry, fix any $u,v\in T$. Since $T_{\lceil\log\log n\rceil+1}$ is a single vertex, surely there is an index $i$ such that $d_{T_i}(u,v)\neq d_{T_{i+1}}(u,v)$, and by the first property of \lemmaref{lem:embed} we get that $\|f_i(u)-f_i(v)\|_\infty= d_T(u,v)$ for the first such index $i$. Note that each $f_j$ is non-expanding, therefore so is $f$, and the isometry follows.

\paragraph{Prioritized dimension.} Recall that the embedding $f_i$ applied with the terminal set $K_i\cup\{z\}$ has $O(\log|K_i|)=O(2^i)$ coordinates. Consider the vertex with priority $j$, $v_j\in T$, and let $i$ be such that $v_j\in K_i$. We have that $i\le\log\log j + 1$. Observe that for every $i'>i$, since $v_j\in K_i$ we have $f_{i'}(v)=f_{i'}(z)=0$. In particular, the total number of nonzero coordinates for $v_j$ in $f$ is at most
\[
\sum_{i'=0}^iO(2^{i'}) = O(2^i) = O(\log j)~.
\]

\end{proof}

\section{Prioritized Embeddings of Arbitrary Metrics into $\ell_\infty$}\label{sec:mat-prior}

In this section we devise prioritized counterparts of the embeddings of Matousek \cite{Mat96b} into $\ell_\infty$, and prioritized the dimension and distortion of these embeddings. Matousek showed that for any $k\ge 1$, any $n$-point metric embeds into $\ell_\infty$ with distortion $2k-1$ and dimension $O(k\cdot n^{1/k}\cdot\log n)$. We will show an adaptation of this embedding that prioritizes the distortion, and a more involved construction that prioritizes the dimension (and slightly the distortion).

We note that at the cost of increasing the distortion by a constant factor, one can use the framework introduced in \cite{EFN15} to obtain prioritized distortion $4\left\lceil\frac{k\cdot\log j}{\log n}\right\rceil-1$ and dimension $O(k^2\cdot n^{1/k}\cdot\log n)$. This is achieved by defining $K_i=\{x_j~:~j\le n^{i/k}\}$, and applying as a black-box the embedding of Matousek on each $K_i$ with distortion $2i-1$ and dimension $O(k\cdot (n^{i/k})^{1/i}\cdot\log n^{i/k})\le O(k\cdot n^{1/k}\cdot\log n)$ for any $x_j\in K_i$. Then extend the embedding of each $K_i$ to the entire set $X$, obtaining distortion $4i-1\le 4\left\lceil\frac{k\cdot\log j}{\log n}\right\rceil-1$ for pairs containing $x_j\in K_i$.
Alternatively, one may use the equivalence between coarse scaling distortion and prioritized distortion \cite{BFN16}, and a result from \cite{ABN06}, to derive prioritized distortion $O\left(\left\lceil\frac{k\cdot\log j}{\log n}\right\rceil\right)$ with improved dimension $O(n^{1/k}\cdot\log n)$, but the leading constant in the distortion is large. However, we would like to obtain an lossless prioritized result, that does not increase the worst case distortion at all. Moreover, such techniques cannot provide prioritized dimension.

\subsection{Prioritized Distortion}\label{sec:prior-dist}

Let $(X,d)$ be a metric space with $|X|=n$, and let $(x_1,\dots,x_n)$ be the given priority. Let $k\ge 1$ be an integer parameter. For $1\le i\le k$ define $S_i=\{x_j~:~ n^{(i-1)/k}< j\le n^{i/k}\}$ (and put $x_1$ in $S_1$).

\paragraph{Construction.} We introduce many "copies" of high priority points, specifically, each point in $S_i$ will have $(2^k\cdot n)^{1-i/k}$ copies. This gives us a (semi)-metric $(X',d')$, by defining the distance between a point to all its copies as 0. Since there are at most $n^{i/k}$ points in $S_i$, the total number of points in $X'$, denoted $N=|X'|$, is at most
\begin{equation}\label{eq:N}
N\le \sum_{i=1}^k(2^k\cdot n)^{1-i/k}\cdot n^{i/k}= 2^k\cdot n\cdot\sum_{i=1}^k 2^{-i}\le 2^k\cdot n~.
\end{equation}

For each $1\le i\le k$ we sample $m=c\cdot N^{1/k}\cdot\ln N$ sets $A_1^{(i)},\dots, A_m^{(i)}$, where $c$ is a constant to be determined later. For every $1\le h\le m$, each element of $X'$ is sampled to $A_h^{(i)}$ independently with probability $N^{-i/k}$. The embedding $f:X\to \R^{km}$ is defined by
\[
f(x)=\bigoplus_{i=1}^k\bigoplus_{h=1}^md(x,A_h^{(i)})~.
\]

\paragraph{Analysis.}
By the triangle inequality, we have that the embedding is non-expansive in each coordinate, thus in $\ell_\infty$ it is non-expansive. It remains to show that, with high probability, for every $1\le j\le n$, any pair containing $x_j$ has distortion at most $2\left\lceil\frac{k\cdot\log j}{\log n}\right\rceil-1$. Fix such a pair $x_j,y\in X$. Let $i$ be such that $x_j\in S_i$, and define $r=\frac{d(x_j,y)}{2i-1}$. Consider the balls $B_0=B(x_j,0)$, $B_1=B^o(y,r)$, $B_2=B(x_j,2r)$, and so on until $B_i$. The balls of even index are closed and centered at $x_j$, while those of odd index are open and centered at $y$. The radius of $B_b$, $0\le b\le i$, is $b\cdot r$. By the definition of $r$, $B_{b-1}\cap B_b=\emptyset$.

Let $1\le b\le i$ be an integer such that $|B_{b-1}|\ge N^{1-(i-b+1)/k}$ and $|B_b|\le N^{1-(i-b)/k}$. Such a $b$ must exist since $x_j\in S_i$ has $(2^k\cdot n)^{1-i/k}$ copies, and
\[
|B_0|= (2^k\cdot n)^{1-i/k}\stackrel{\eqref{eq:N}}{\ge} N^{1-i/k}~,
\]
so if no such $b$ exists, it follows that $|B_1|> N^{1-(i-1)/k}$, and so on until  $|B_i|>N^{1-(i-i)/k}=N$, a contradiction.

For each $1\le h\le m$, let ${\cal E}_h^{hit}$ be the event that $B_{b-1}\cap A_h^{(k-i+b)}\neq\emptyset$, and let ${\cal E}_h^{miss}$ be the event that $B_{b}\cap A_h^{(k-i+b)}=\emptyset$. Since the balls are disjoint and each point joins the $\{A_h^{(i)}\}$ sets independently, these events are independent.  We calculate:
\[
\Pr[{\cal E}_h^{hit}]=1-\left(1-N^{-(k-i+b)/k}\right)^{|B_{b-1}|}\ge 1-e^{-N^{-(k-i+b)/k}\cdot N^{1-(i-b+1)/k}}\ge N^{-1/k}/2~.
\]
\[
\Pr[{\cal E}_h^{miss}]=\left(1-N^{-(k-i+b)/k}\right)^{|B_b|}\ge \left(1-N^{-1+(i-b)/k}\right)^{N^{1-(i-b)/k}}\ge 1/4~.
\]
Thus, the probability that there is no $1\le h\le m$ for which both events occur, is bounded by
\[
\left(1-N^{-1/k}/8\right)^m\le e^{-N^{-1/k}/8\cdot c\cdot N^{1/k}\ln N}=1/N^2~,
\]
whenever $c=16$. If there is an $1\le h\le m$ such that both events occur, then for odd $b$ we have $d(x_j,A_h^{(k-i+b)})\le (b-1)r$ and $d(y,A_h^{(k-i+b)})\ge br$ (for even $b$ replace the roles of $x_j,y$). It follows that
\[
|d(y,A_h^{(k-i+b)})-d(x_j,A_h^{(k-i+b)})|\ge br-(b-1)r = r =\frac{d(x_j,y)}{2i-1}~.
\]
Recall that $x_j\in S_i$ means that $n^{(i-1)/k}< j\le n^{i/k}$, or equivalently, $i-1< k\cdot\frac{\log j}{\log n}\le i$. So the prioritized distortion is indeed at most $2\left\lceil\frac{k\cdot\log j}{\log n}\right\rceil-1$. The dimension we obtain is $k\cdot m=O(k\cdot N^{1/k}\cdot\ln N)=O(k\cdot n^{1/k}\cdot\ln n)$. This is since $N^{1/k}=(2^k\cdot n)^{1/k}=2n^{1/k}$, and $\ln N=\ln(2^k\cdot n) =O(k+\ln n)=O(\ln n)$, using that it makes no sense to take $k>\log n$.

We proved the following.

\begin{theorem}\label{thm:matousek-dist}
For any parameter $k\ge 1$, any $n$-point metric embeds into $\ell_\infty$ with prioritized distortion $2\left\lceil\frac{k\cdot\log j}{\log n}\right\rceil-1$ and dimension $O(k\cdot n^{1/k}\cdot\log n)$.
\end{theorem}

\subsection{Prioritized Dimension}

Note that in the construction of \sectionref{sec:prior-dist}, a point in $S_i$ has $(2^k\cdot n)^{1-i/k}$ copies, which intuitively means it is very unlikely that all its copies are missed by the a set $A_h^{(b)}$ for $b<k-i$ (that includes each point independently with probability $N^{-b/k}$). If a set hits at least one copy, it means that such a point will have a value of zero in the corresponding coordinate (the distance to that set is zero). However, it might get a high number of nonzeros from the sets with larger value of $b$. Recall that we repeat $O(N^{1/k}\ln N)$ times the random choices for each density $1\le b\le k$, then in particular, the last density level $b=k$ is very likely to incur high dimension for every point: all the copies of $x_1$, say, are expected to be missed by $\Omega(N^{1/k})$ of the sets in $\{A_h^{(k)}\}_h$ (and thus have many nonzero coordinates).

To alleviate this issue, we refine the sampling probabilities, and enforce exponentially smaller gaps between neighboring densities, as the exponent becomes closer to $-1$. Specifically, we will define the sets $S_i$ for every $0\le i<\log\log n$ by
\[
S_i=\{x_j~:~2^{2^i}<j\le 2^{2^{i+1}}\}~.
\]
Note that $x_1,x_2$ do not belong to any of these sets. To handle this, simply introduce additional $2$ coordinates, and map each $x\in X$ to the vector $(d(x,x_1),d(x,x_2))$. Clearly this is non-expanding, and has distortion 1 for pairs containing $x_j$ with $1\le j\le 2$. So in all that follows we only care for pairs containing $x_j$ for $j>2$.

Now, each point $x\in S_i$ will have $C(i)=\frac{n(\log\log n+1)^2}{2^{2^{i+1}}\cdot (i+2)^2}$ copies. Observe that $C(i)\ge 1$ for all $i<\log\log n$ (assuming $\log\log n$ is an integer).
Since $|S_i|\le 2^{2^{i+1}}$, we have that 
the total number of points in the new metric $(X',d')$ is
\[
N\le \sum_{i=0}^{\log\log n}C(i)\cdot 2^{2^{i+1}}=n(\log\log n+1)^2\cdot\sum_{i=0}^{\log\log n}\frac{1}{(i+2)^2}\le n(\log\log n+1)^2~.
\]
For every $0\le i\le\log\log n$ we define $R(i)=c\cdot 2^{(2^i+2)/k}\cdot\ln n$ as the number of samples to be taken with the appropriate density corresponding to $i$. Note that when $i=\log\log n$ we have $R(i)= \Theta(n^{1/k}\cdot\ln n)$ as in the previous section, but for small values of $i$, e.g. $i<\log k$, we have $R(i)\le 2c\cdot\ln n$.

We now define the embedding. For every $0\le i<\log\log n$ and $1\le s\le k$ sample $R(i)$ sets $\{A_h^{(s,i)}\}_{1\le h\le R(i)}$, so that every element in $X'$ is included in $A_h^{(s,i)}$ independently with probability \\ $\min\{\frac{2^{2^i(1+s/k)-2+2s/k}\cdot(i+2)^2}{N},1\}$.
To compensate for the increased number of points $N$, we also introduce additional $ c\cdot \ln n$ sets $\{E_g\}$ for every $1\le g\le c\cdot\ln n$, where each set $E_g$ includes every element of $X'$ independently with probability $\frac{1}{N}$.
The embedding $f:X\to\R^{O(k\cdot n^{1/k}\cdot\ln n)}$ is defined for every $x\in X$ as
\[
f(x)=\bigoplus_{s=1}^k\bigoplus_{i=0}^{\log\log n-1}\bigoplus_{h=1}^{R(i)}d(x,A_h^{(s,i)})\oplus\bigoplus_{g=1}^{c\cdot\ln n}d(x,E_g)~.
\]

\paragraph{Distortion bound.}
Note that $f$ is non-expansive, and we bound its contraction. Fix a pair $x_j,y\in X$ for $2< j\le n$, and let $0\le i< \log\log n$ be such that $x_j\in S_i$.
Define $\alpha=2ki+1$, $q=(\alpha+1)/2=ki+1$, and $r=\frac{d(x_j,y)}{\alpha}$.
Define $B_0=B(x_j,0)$, $B_1=B^o(y,r)$, $B_2=B(x_j,2r)$, and so on until $B_q$, as before.
Note that since $x_j\in S_i$, we have that $|B_0|=C(i)=\frac{n(\log\log n+1)^2}{2^{2^{i+1}}\cdot (i+2)^2}\ge \frac{N}{2^{2^{i+1}}\cdot (i+2)^2}$.

If there exists $1\le b\le k$ such that
\[
|B_{b-1}|\ge \frac{N\cdot 2^{(2^i+2)\cdot (b-1)/k}}{2^{2^{i+1}}\cdot (i+2)^2}~~~ {\rm and}~~~ |B_b|\le \frac{N\cdot 2^{(2^i+2)\cdot b/k}}{2^{2^{i+1}}\cdot (i+2)^2}~,
\]
then we will stop here. Otherwise, we have
\[
|B_k|> \frac{N\cdot 2^{(2^i+2)\cdot k/k}}{2^{2^{i+1}}\cdot (i+2)^2}\ge\frac{N}{2^{2^i}\cdot (i+1)^2}
\]
(using that $\frac{4}{(i+2)^2}\ge \frac{1}{(i+1)^2}$ holds for all $i\ge 0$). We then continue to ask whether there is a $1\le b\le k$ so that
\[
|B_{k+b-1}|\ge \frac{N\cdot 2^{(2^{i-1}+2)\cdot (b-1)/k}}{2^{2^i}\cdot (i+1)^2}~~~ {\rm and}~~~ |B_{k+b}|\le \frac{N\cdot 2^{(2^{i-1}+2)\cdot b/k}}{2^{2^i}\cdot (i+1)^2}~,
\]
if there is no such $b$, we get that
\[
|B_{2k}|> \frac{N\cdot 2^{(2^{i-1}+2)\cdot k/k}}{2^{2^i}\cdot (i+1)^2}\ge\frac{N}{2^{2^{i-1}}\cdot i^2}~.
\]
In general, after $\ell$ such phases we have that $|B_{\ell k}|>\frac{N}{2^{2^{i-\ell+1}}\cdot (i-\ell+2)^2}$. If all these phases fail, then $|B_{ik}|\ge N/16$. (Recall that $q>ki$, so all these balls do exist.)

Let us first see what happens if there is an $0\le \ell< i$ and $1\le b\le k$ such that
\begin{eqnarray*}
&|B_{\ell k+b-1}|\ge \frac{N\cdot 2^{(2^{i-\ell}+2)\cdot (b-1)/k}}{2^{2^{i-\ell+1}}\cdot (i-\ell+2)^2}=\frac{N}{2^{2^{i-\ell}(2-(b-1)/k)-2(b-1)/k}\cdot(i-\ell+2)^2}~~~ {\rm and}~~~\\& |B_{\ell k+b}|\le \frac{N\cdot 2^{(2^{i-\ell}+2)\cdot b/k}}{2^{2^{i-\ell+1}}\cdot (i-\ell+2)^2}=\frac{N}{2^{2^{i-\ell}(2-b/k)-2b/k}\cdot(i-\ell+2)^2}~.
\end{eqnarray*}
For each $1\le h\le R(i-\ell)$, let ${\cal E}_h^{hit}$ be the event that $B_{\ell k+b-1}\cap A_h^{(k-b,i-\ell)}\neq\emptyset$, and let ${\cal E}_h^{miss}$ be the event that $B_{\ell k+b}\cap A_h^{(k-b,i-\ell)}=\emptyset$.
Recall that $A_h^{(k-b,i-\ell)}$ contains each element of $X'$ independently with probability $$\frac{2^{2^{i-\ell}(1+(k-b)/k)-2+2(k-b)/k}\cdot(i-\ell+2)^2}{N}=\frac{2^{2^{i-\ell}(2-b/k)-2b/k}\cdot(i-\ell+2)^2}{N}~.$$
It follows that
\[
\Pr[{\cal E}_h^{hit}]=1-\left(1-\frac{2^{2^{i-\ell}(2-b/k)-2b/k}\cdot(i-\ell+2)^2}{N}\right)^{|B_{b-1}|}\ge 1-e^{-2^{-(2^{i-\ell}+2)/k}}\ge 2^{-(2^{i-\ell}+2)/k-1}~.
\]
\[
\Pr[{\cal E}_h^{miss}]=\left(1-\frac{2^{2^{i-\ell}(2-b/k)-2b/k}\cdot(i-\ell+2)^2}{N}\right)^{|B_b|}\ge 1/4~.
\]
Recall that $R(i-\ell)=c\cdot 2^{(2^{i-\ell}+2)/k}\cdot\ln n$, so the probability that none of the $1\le h\le R(i-\ell)$ have that both events occur is at most
\[
(1-2^{-(2^{i-\ell}+2)/k-3})^{c\cdot 2^{(2^{i-\ell}+2)/k}\cdot\ln n}\le 1/n^2~,
\]
whenever $c=16$.
Note that if there is an $h$ for which both events occur, it implies that
\[
|d(x_j,A_h^{(k-b,i-\ell)})-d(y,A_h^{(k-b,i-\ell)})|\ge r~,
\]
and thus the distortion for the pair $x_j,y$ is at most $\alpha$.

We now consider the case that no such $\ell$ and $b$ were found, so we have the guarantee that $|B_{ki}|\ge N/16$.
Recall that $B_q=B_{ki+1}$, and $|B_q|\le N$.
For each $1\le g\le c\ln n$, let ${\cal F}_g^{hit}$ be the event that $B_{ki}\cap E_g\neq\emptyset$, and let ${\cal F}_g^{miss}$ be the event that $B_{ki+1}\cap E_g=\emptyset$. Since the balls are disjoint these events are independent, and 
\[
\Pr[{\cal F}_g^{hit}]=1-\left(1-\frac{1}{N}\right)^{|B_{ki}|}\ge 1-e^{-N/(16N)}\ge \frac{1}{32}~.
\]
\[
\Pr[{\cal F}_g^{miss}]=\left(1-\frac{1}{N}\right)^{|B_{ki+1}|}\ge 1/4~.
\]
Thus, the probability that there is no $1\le g\le c\ln n$ for which both events occur, is bounded by
\[
\left(1-\frac{1}{128}\right)^{c\ln n}\le 1/n^2~,
\]
whenever $c$ is large enough. If there is an $1\le g\le c\ln n$ such that both events occur, then the distortion of the pair $x_j,y$ is bounded by $\alpha$ in this case as well. By the union bound, there is a constant probability (which can easily be made polynomially close to 1 by increasing $c$), that all pairs have such bounded distortion.
As $x_j\in S_i$ we have that $i\le \log\log j$, so the distortion for pairs containing $x_j$ is at most $\alpha =2k\log\log j+1$.

\paragraph{Bounding the dimension.} We now turn to proving a bound on the number of nonzero coordinates of $x_j$.
Let $0\le i< \log\log n$ be such that $x_j\in S_i$. Recall that $x_j$ has $C(i)=\frac{n(\log\log n+1)^2}{2^{2^{i+1}}\cdot (i+2)^2}\ge\frac{N}{2^{2^{i+1}}\cdot (i+2)^2}$ copies, and each element in chosen to be in $A_h^{(s,i)}$ independently with probability $\frac{2^{2^i(1+s/k)-2+2s/k}\cdot(i+2)^2}{N}$, where $1\le h\le R(i)$ and $R(i)=O(2^{(2^i+2)/k}\cdot\ln n)$.

For the first $k(i+1)$ collections of sets $\{A_h^{(s,t)}\}_h$ with $0\le t\le i$ and $1\le s\le k$, we will not try to prove that they contain a copy of $x_j$. The total number of such sets is
\begin{equation}\label{eq:reew}
k\cdot \sum_{t=0}^{i}R(t) = O(k\cdot\ln n\cdot (2^{(2^i+2)/k}+\log k))=O(k\cdot\ln n\cdot( j^{2/k}+\log k)).
\end{equation}
(The additive term of $\log k$ comes from the first $\log k$ terms in the summation, the rest is dominated by the last term.)
We also have an additional $c\cdot \ln n$ coordinates corresponding to the sets $\{E_g\}$, so the total number of coordinates so far is still as in \eqref{eq:reew}.

Consider the set $A_h^{(s,t)}$ for some $t\ge i+1$, and let $Y_h^{(s,t)}$ be an indicator random variable for the event that $A_h^{(s,t)}$ does not contain any copy of $x_j$. Observe that $x_j$ will have nonzero value in the coordinate of $A_h^{(s,t)}$ iff $Y_h^{(s,t)}=1$. Denote $Y^{(s,t)}=\sum_{h=1}^{R(t)}Y_h^{(s,t)}$, and $Y=\sum_{t\ge i+2}\sum_{s=1}^k Y^{(s,t)}$. The number of nonzero coordinates of $x_j$ from the sets $A_h^{(s,t)}$, for $t\ge i+1$, is equal to $Y$.

We will soon show that $\E[Y]\le c'\cdot k\cdot \ln n$ for some constant $c'$. Thus, by Chernoff bound (note the random variables $Y_h^{(s,t)}$ are independent)
\[
\Pr[Y>5c'k\cdot\ln n]\le e^{-2k\ln n}\le 1/n^2~.
\]
We will conclude that with constant probability, every $x_j$ has at most $O(k\cdot (j^{2/k}+\log k)\cdot\ln n)$ nonzero coordinates.

For $t\ge i+1$ we have that
\[
\E[Y_h^{(s,t)}]=\left(1-\frac{2^{2^t(1+s/k)-2+2s/k}\cdot(t+2)^2}{N}\right)^{\frac{N}{2^{2^{i+1}}\cdot (i+2)^2}}\le e^{-2^{2^t/k+2^t-2^{i+1}}}~.
\]
Then
\[
\E[Y^{(s,t)}]=\sum_{h=1}^{R(t)}\E[Y_h^{(s,t)}]\le O(2^{(2^t+2)/k}\cdot\ln n)\cdot e^{-2^{2^t/k+2^t-2^{i+1}}}~,
\]
and
\[
\E[Y]=\sum_{t\ge i+1}\sum_{s=1}^k\E[Y^{(s,t)}]\le O(k\cdot\ln n)\cdot \sum_{t\ge i+1}2^{(2^t+2)/k}\cdot e^{-2^{2^t/k+2^t-2^{i+1}}}=O(k\cdot\ln n)~,
\]
as the sum converges to a small constant.
We proved the following.

\begin{theorem}\label{thm:matousek-dim}
For any parameter $k\ge 1$, any $n$-point metric embeds into $\ell_\infty$ with prioritized distortion $2k\log\log j+1$ and prioritized dimension $O(k\cdot (j^{2/k}+\log k)\cdot\ln n)$.
\end{theorem}


\section{Prioritized Embedding into a Single Tree}

Define $\Phi$ to be the family of functions $\alpha:\N\to\R_+$ that satisfy the following properties:
\begin{itemize}
\item $\alpha$ is non-decreasing.
\item $\sum_{i=1}^\infty 1/\alpha(i)< 1$.
\end{itemize}
For instance, one can take $\alpha(j)=c\cdot j\cdot \log j\cdot(\log\log j)^{1.1}$ for a suitable constant $c$.

Recall the result of \cite{EFN15}, who proved that every metric space embeds into a single tree with prioritized distortion $O(\alpha(j))$, for any function $\alpha\in \Phi$. It was also shown in \cite{EFN15} that for every $\alpha\notin\Phi$, there exists a metric that does not admit a prioritized embedding into a single tree with distortion less than $\alpha/8$.

Here we extend the result of \cite{EFN15}, by showing a prioritized embedding of any metric into an ultrametric (which is a special (and useful) kind of tree metric), and of a graph into one of its spanning tree, with prioritized distortion $O(\alpha(j))$ for any $\alpha\in\Phi$.

\subsection{Single Ultrametric}\label{sec:single-ultra}

An ultrametric $\left(U,d\right)$ is a metric space satisfying a
strong form of the triangle inequality, that is, for all $x,y,z\in U$,
$d(x,z)\le\max\left\{ d(x,y),d(y,z)\right\} $. The following definition
is known to be an equivalent one (see \cite{BLMN03}).
\begin{definition}\label{def:ultra}
An ultrametric $U$ is a metric space $\left(U,d\right)$ whose elements
are the leaves of a rooted labeled tree $T$. Each $z\in T$ is associated
with a label $\ell\left(z\right)\ge0$ such that if $q\in T$ is a
descendant of $z$ then $\ell\left(q\right)\le\ell\left(z\right)$
and $\ell\left(q\right)=0$ iff $q$ is a leaf. The distance between
leaves $z,q\in U$ is defined as $d_{T}(z,q)=\ell\left(\mbox{lca}\left(z,q\right)\right)$
where $\mbox{lca}\left(z,q\right)$ is the least common ancestor of
$z$ and $q$ in $T$.
\end{definition}

\begin{theorem}\label{thm:single-ultra}
For any finite metric space $(X,d)$ and any $\alpha\in \Phi$, there is a (non-contractive) embedding of $X$ into an ultrametric with priority distortion $2\alpha(j)$.
\end{theorem}
\begin{proof}
We create the ultrametric tree $T$ for $X$ by partitioning $X$ into $X_1$ and $X_2$, creating a root for $T$ labeled by $\Delta=\diam(X)$, and placing as its two children the trees $T_1$ and $T_2$ created recursively for $X_1$ and $X_2$ respectively. Clearly this embedding is non-contractive (as the distance between two points in $T$ is the diameter of the cluster in which they are separated). It remains to see how to generate a partition of $X$ that will have the required prioritized expansion. Note that a pair separated by the initial partition will have their distance exactly $\Delta$ in the ultrametric. This means that for each $j\in [n]$, we cannot separate $x_j$ from $x_i$ whenever
\begin{equation}\label{eq:bad}
2\min\{\alpha(i),\alpha(j)\}\cdot d(x_j,x_i)<\Delta~.
\end{equation}
Call a pair $x_j,x_i$ {\em bad} with respect to $X_1$ if $|\{x_j,x_i\}\cap X_1|=1$ and \eqref{eq:bad} holds.
We are now ready to define the partition of $X$:
Pick any pair $u,v\in X$ so that $d(u,v)=\Delta$. Initialize $r=0$, and repeat:
\begin{enumerate}
\item Set $X_1\leftarrow B(u,r)$.
\item Take the minimal $j$ such that there exists a bad pair (with respect to $X_1$) containing $x_j$. If there is no bad pair, stop.
\item Set $r\leftarrow r+\Delta/(2\alpha(j))$. Return to 1.
\end{enumerate}
We claim that at the end of the process $v\notin X_1$. To see this, we first argue that for each $1\le j\le n$, there can be at most two iterations in which $r$ increases by $\Delta/(2\alpha(j))$. This is because after the first such iteration, we claim that $x_j\in X_1$. This is because either $x_j\in X_1$ even before the radius increase, otherwise, it must be that $x_i\in X_1$, and the radius increases by $\Delta/(2\alpha(j))>d(x_i,x_j)$, so by the triangle inequality $x_j$ will be in $X_1$. After the second increase we have that $B(x_j,\Delta/(2\alpha(j)))\subseteq X_1$ by the minimality of $j$, so there are no more bad pairs $x_j,x_i$ with $j<i$ (since $\alpha$ is monotone). Using that $\alpha\in\Phi$, we obtain that the total increase in the radius is at most
\begin{equation}\label{eq:sum}
\sum_{j=1}^n\frac{\Delta}{\alpha(j)}<\Delta~.
\end{equation}
We conclude that the partition of $X$ to $X_1$ and $X_2=X\setminus X_1$ is non-trivial ($u\in X_1$ and $v\in X_2$). Furthermore, for every pair $x,y\in X$ separated by the partition, $x\in X_1$ and $y\in X_2$ must abide the required distortion, since this pair is not bad, thus \eqref{eq:bad} does not hold.

The proof that the recursive construction satisfies the distortion constraints is by induction on $n=|X|$. Indeed both $|X_1|,|X_2|<|X|$, so there are trees $T_1$ and $T_2$ that preserve all the distortion constraints among the pairs in ${X_1\choose 2}\cup{X_2\choose 2}$ (note that $\alpha$ is monotone, so using the induced priority ranking on $X_1$ and $X_2$ will only yield improved distortion bounds). Finally, all pairs in $X_1\times X_2$ suffer appropriate distortion by the discussion above.
\end{proof}

\subsection{Single Spanning Tree}

In this section we extend the partition technique used to prove \theoremref{thm:single-ultra}, and obtain an embedding of any graph into a single spanning tree with priority distortion.
\begin{theorem}\label{thm:single-spanning}
For any weighted graph $G=(V,E)$ and any $\alpha\in \Phi$, there is an embedding of $G$ into a spanning tree with priority distortion $O(\alpha(j))$.
\end{theorem}

\paragraph{Petal Decomposition}

We will make use of a construction of a spanning tree by a {\em Petal-Decomposition}, introduced in \cite{AN12}. This decomposition generates a hierarchical partition of the vertex set of a graph, which can be viewed as a laminar family over subsets of $V$.
Every cluster $X$ in the family is also associated with a center $x\in X$ and a target $t\in X$. The cluster $X$ is partitioned to $X_0,X_1,\dots,X_s$ (for some integer $s$), and the algorithm also specifies $s$ edges that connect these clusters in a tree manner. Each $X_i$ for $1\le i\le s$ is called a {\em petal}, and is generated in the graph induced on $Y_{i-1}=X\setminus(\cup_{j=1}^{i-1}X_j)$, by picking an arbitrary target $t_i$ (except perhaps for $t_1$, which is picked on the shortest path from $x$ to $t$), of distance at least $3\rad(X)/4$ from $x_0$,\footnote{The radius of a cluster $X$ with center $x$ is defined as $\rad(X)=\max_{y\in X}d(x,y)$. Observe that $\diam(X)/2\le\rad(X)\le\diam(X)$.} picking a radius $0\le r_i\le \rad(X)/8$, and setting $X_i=P(t_i,r_i)$. The final cluster $X_0=Y_s$, is simply the remaining vertices when there are no vertices sufficiently far from the center $x$.

We now explain how a petal $P(t,r)$ is defined.
Given a weighted undirected graph $G=(V,E,w)$ with a center $x$ and a desired target $t$ for the new petal, let $P_{tx}$ be the shortest path from $t$ to $x$. Let $\tilde{G}=(V,A,\bar{w})$ be the weighted directed graph created by adding the two directed edges $(u,v), (v,u) \in A$ for each $\{u,v\} \in E$, and setting $\bar{w}(u,v)=d_G(u,v)-(d_G(v,x)-d_G(u,x))$. However, the weight of each directed edge $(u,v)$ on the path $P_{tx}$, where $u$ is closer to $t$, is set to be $\bar{w}(u,v)=w(u,v)/2$ (note that by our definition it should have been $\bar{w}(u,v)=2w(u,v)$).
The Petal $P(t,r)$ is the ball around $t$ of radius $r/2$ in $\tilde{G}$ (the center of the new petal is defined as the farthest point from $t$ on $P_{tx}$ which lies in $P(t,r)$). The crucial properties of the Petal-Decomposition are (see \cite{AN12} for a proof):
\begin{enumerate}
\item For each cluster $X$, the tree $T[X]$ spanning $X$ has radius at most $4\rad(X)$.
\item For any $y\in V$ within distance $\delta$ from (some vertex in) the petal $P(t,r)$, it holds that $y\in P(t,r+4\delta)$.
\end{enumerate}
The first property mean that like in the ultrametric case, a separated pair in a cluster of radius $\Delta$ will be at distance $O(\Delta)$ in the final tree, and the second property suggests that petals are similar to balls.

\begin{proof}[Proof of Theorem~\ref{thm:single-spanning}]
We generate a tree $T$ by applying Petal-Decomposition on the graph. For each cluster $X$ generated in the process, that is partitioned into $X_0,X_1,\dots,X_s$, we must provide for every $1\le i\le s$ a radius from the range $[0,\rad(X)/8]$. Note that if a pair $x,y\in X$ is separated when creating $X_i$ (that is $x\in X_i$ and $y\in Y_i$), then by the first property their final distance in the tree will be at most $\diam(T[X])\le 8\rad(X)$. It follows that to satisfy the distortion constraints we can use essentially the same definition of bad pairs and algorithm for choosing $r$ as in \sectionref{sec:single-ultra}.
The only difference is that we will use larger constants: a pair $x_j,x_{j'}$ is called {\em bad} if they are separated when partitioning $X$ and $128\min\{\alpha(j),\alpha(j')\}\cdot d(x_j,x_{j'})<\rad(X)$. The process to find the radius $r_i$ for carving the next petal $X_i$ with a given target $t_i$ out of the graph induced on $Y_{i-1}$ is the following. Initialize $r_i=0$, and repeat:
\begin{enumerate}
\item Set $X_i\leftarrow P(t_i,r_i)$.
\item Take the minimal $j$ such that there exists a bad pair containing $x_j$. If there are no bad pairs, stop.
\item Set $r_i\leftarrow r_i+\rad(X)/(16\alpha(j))$. Return to 1.
\end{enumerate}
Using the second property of petals, we obtain that $j$ can appear as the minimal index at most twice. This is because after the first time, as the radius increased by $\diam(X)/(16\alpha(j))$, so vertices within distance $\rad(X)/(64\alpha(j))\ge d(x_{j'},x_j)$ will be contained in $X_i$ (where $j'>j$ is index of the other vertex in the bad pair). In particular, it follows that after the first time it must be that $x_j\in X_i$. After the second time $j$ is the minimal we get that $B(x_j,\rad(X)/(64\alpha(j)))\subseteq X_i$. Thus (by a similar argument to the section above) $x_j$ will never be the minimal of a bad pair again (until $r_i$ is set). By the calculation done as in \eqref{eq:sum}, the total increase in $r_i$ is at most $2\rad(X)/16=\rad(X)/8$, as required.

Finally, we bound the distortion of any pair $x_j,x_{j'}$ which is not bad. If this pair is separated while partitioning $X$, their distance in $T$ will be at most $8\rad(X)$, so by definition of a bad pair, the distortion they suffer is at most $8\cdot 128\alpha(j)=O(\alpha(j))$.

\end{proof}
\bibliographystyle{alpha}
\bibliography{art}

\end{document}